\documentclass[sigconf]{acmart}

\AtBeginDocument{%
  \providecommand\BibTeX{{%
    \normalfont B\kern-0.5em{\scshape i\kern-0.25em b}\kern-0.8em\TeX}}}

\setcopyright{acmcopyright}
\copyrightyear{2018}
\acmYear{2018}
\acmDOI{10.1145/1122445.1122456}

\acmConference[Woodstock '18]{Woodstock '18: ACM Symposium on Neural
  Gaze Detection}{June 03--05, 2018}{Woodstock, NY}
\acmBooktitle{Woodstock '18: ACM Symposium on Neural Gaze Detection,
  June 03--05, 2018, Woodstock, NY}
\acmPrice{15.00}
\acmISBN{978-1-4503-XXXX-X/18/06}

\usepackage{balance} 
\usepackage[linesnumbered,ruled,vlined]{algorithm2e}
\usepackage{amsmath}
\usepackage{makecell}
\usepackage{soul}
\usepackage{enumitem}
\usepackage{caption}
\usepackage{subcaption}
\usepackage{pifont}
\newcommand{\xmark}{\ding{55}}
\newcommand{\cmark}{\ding{51}}
\newtheorem{obs}{Observation}
\usepackage[]{algorithm2e}
\usepackage{footnote}
\makesavenoteenv{tabular}
\makesavenoteenv{table}
\usepackage{amsfonts}
\usepackage{multirow}
\begin{document}

\title{Mechanism Design without Money for Fair Allocations}

\author{Manisha Padala and Sujit Gujar}
\email{manisha.padala@research.iiit.ac.in}
\email{sujit.gujar@iiit.ac.in}
\affiliation{%
  \institution{Machine Learning Lab, International Institute of Technology, Hyderabad}
  \country{India}
}

\renewcommand{\shortauthors}{Padala and Gujar}

\begin{abstract}
Fairness is well studied in the context of resource allocation. Researchers have proposed various fairness notions like envy-freeness (EF), and its relaxations, proportionality and max-min share (MMS). There is vast literature on the existential and computational aspects of such notions. While computing fair allocations, any algorithm assumes agents' truthful reporting of their valuations towards the resources. Whereas in real-world web-based applications for fair division, the agents involved are strategic and may manipulate for individual utility gain. In this paper, we study strategy-proof mechanisms without monetary transfer, which satisfies the various fairness criteria.

We know that for additive valuations, designing truthful mechanisms for EF, MMS and proportionality is impossible. Here we show that there cannot be a truthful mechanism for EFX and the existing algorithms for EF1 are manipulable. We then study the particular case of single-minded agents. For this case, we provide a Serial Dictatorship Mechanism that is strategy-proof and satisfies all the fairness criteria except EF.
\end{abstract}

%




\maketitle

\section{Introduction}
 
Fair division of resources is critical in various situations like division of inheritance or land and allocation of rooms to housemates, jobs to workers, time slots to courses. In a typical scenario, the agents involved report their valuations for the resources available. The central aggregator or the underlying software aggregates these reported valuations to output a fair allocation. Various web-based applications like Spliddit \footnote{www.spliddit.org}, Fair Proposals System \footnote{www.fairproposals.com}, Coursematch \footnote{www.coursematch.io}, Divide Your Rent Fairly \footnote{https://www.nytimes.com/interactive/2014/science/rent-division-calculator.html}, etc offer such solutions readily. Often the participants are strategic and misreport their valuations to improve their utility.  The party that strictly adheres to the protocol and reveals its true valuation (while it can misreport and achieve more utility) may find it unfair if others misreport for their benefit even though the underlying algorithm is fair for the reported types. In auction settings, one prevents such strategic manipulations through monetary transfers. Whereas in resource allocation, no monetary transfers are allowed. Hence, it is essential to look for truthful mechanisms that ensure fairness without payments.

In this paper, we focus on indivisible resources and the fairness notions of \emph{envy freeness} (EF), \emph{proportionality} and \emph{maxi-min share} (MMS).  Proportionality \cite{steihaus1948} is the first concept of fairness ever proposed. It ensures that each agent receives a fair share of its utility. Another popular notion is envy-freeness (EF). An allocation is EF when no pair of agents exist such that one of the agents increases its utility by exchanging their allocated goods \cite{foley1967}. For divisible goods, EF allocations always exist \cite{stromquist1980}, and complete allocation may not exist for indivisible goods. It is also NP-hard to compute an approximation to EF \cite{lipton2004}. When the valuations are sub-additive, EF implies proportionality \cite{bouveret16}. Although proportionality is a weaker notion, its existence is still not guaranteed for indivisible goods. 

Given the above results, in \cite{lipton2004, budish2011}, the authors relax EF and introduce EF up to the most-valued good or EF1. An EF1 allocation is always guaranteed to exist even for indivisible goods and can be computed in polynomial time by the cycle-elimination algorithm \cite{lipton2004}. It is interesting to consider EFX, which is EF up to the least-valued good \cite{caragiannis2019}. It is stronger than EF1. EFX always exists for up to three agents \cite{efx3}. For indivisible goods, another fairness criteria considered is MMS \cite{budish2011}, where each agent's utility is at least its MMS guarantee. The MMS guarantee is the worst-case value an agent receives when partitioning the goods and others choose before it. MMS allocation is guaranteed to exist for up to two agents \cite{procaccia2014}.

The above existential and complexity results assume that each agent's preferences (determined using their valuations for each bundle) are known.  In this work, we are interested in preference elicitation to prevent manipulations. Hence we study the existence of truthful or SP (Strategy-Proof) mechanisms that ensure fairness or \emph{Strategy-Proof Fair} (SPF).
\begin{savenotes}
\begin{table*}[!ht]
    \centering
    \renewcommand{\arraystretch}{1.05}{
    \begin{tabular}{|c|l|c|l|l|}
        \hline
        Property & Single-Minded &  Identical  & \multicolumn{2}{c|}{Additive}  \\
        \cline{4-5}
        &&Additive (NSP)&($n=2$)&($n\ge m$)\\
        \hline
        EF              &    &  & \xmark \cite{lipton2004} & \\
        Proportionality &{\color{blue}\cmark (SD)}\footnote{Only when proportionality exists}    &  & \xmark \cite{amanatidis2016} ({\color{blue} alternate proof}) & \\
        EFX             &{\color{blue}\cmark (SD)}&  & {\color{blue}\xmark (Theorem \ref{thm:efx}) (even for $m=4$ )} & \\
        EF1             &{\color{blue}\cmark (SD)}& {\color{blue}\cmark (RSD)} & \xmark \cite{amanatidis2017} ($m\geq5$) & {\color{blue}\cmark (RSD)}   \\
        MMS             &{\color{blue}\cmark (SD)}&  & \xmark \cite{amanatidis2016}   &\\
       \hline
    \end{tabular}}
    \caption{Existence of SPF Mechanisms for Various Types of Valuations}
    \label{tab:mvp}

\end{table*}
\end{savenotes}
\subsection{Strategy-Proof Fair}
A direct-revelation mechanism takes all the input valuation functions and returns an allocation. A direct-revelation mechanism is SPF if it ensures fair allocation when no agent can gain higher utility by misreporting. In mechanism design literature, it is standard to introduce payments to design truthful mechanisms, especially in auction settings \cite{hartline03, edith11, TangZ15}. In this paper, we focus on the basic model of fair mechanism design without money. When the goods are divisible, the authors in \cite{menon17, bei17} prove that no deterministic SP mechanism (without monetary transfers) is proportional or even approximately proportional for complete allocation. Since EF is stronger than proportionality, having an SP mechanism for EF is also impossible. It is known that there exist randomized SP mechanisms which ensure EF when the goods are divisible \cite{mossel2010}. There are other works like \cite{cole2013, chen2013, branzei2017} which give SPF mechanisms without money for divisible goods. In \cite{lang11}, the authors show that sequential allocation is strategy-proof when agents have identical rankings. This way of allocation is referred to as \emph{Picking Sequences}.
\subsubsection*{SPF Mechanism for Indivisible goods}
In \citeauthor{lipton2004} \cite{lipton2004}, the authors prove that it is impossible to design a truthful mechanism that achieves minimum envy or EF by providing a counterexample. When there are two agents ($n=2$) and the number of goods ($m$) are greater than 5, there cannot be a deterministic SP mechanism with complete allocation for EF1 even for additive valuations \cite{amanatidis2017}. There are impossibility results for MMS in \cite{amanatidis2016}, the authors prove that for 2 agents, there is no truthful mechanism that ensures better than $\frac{1}{m/2}$ -MMS allocation.

\subsection{Our Contribution}
Given the above limiting results, we explore the following in this paper
\begin{enumerate}
    \item We study the EFX property for two agents, where it is guaranteed to exist. From \cite{amanatidis2017}, it is impossible to have SP mechanism for EF1 with two agents and more than $5$ goods. EFX being a stronger property also follows the same result for the given setting. This paper provides an example that proves that designing an SP mechanism for EFX is impossible even when the number of goods is $4$.  
    \item Aligning with the results of \cite{amanatidis2017}, we provide examples to show that greedy round-robin algorithm and cycle-elimination algorithm for finding EF1 are manipulable. When agents have identical additive allocations, greedy round robin provides allocations that is EF1 as well as strategy-proof.
    \item Given that the valuations can be very complex to represent in general, we restrict ourselves to the simpler case of (SM) \emph{single-minded} agents. SM bidders is very common in the auction literature multi-item setting \cite{rassenti1982combinatorial, mcmillan1994selling}. In such a setting, we provide (SD) \emph{(Serial Dictatorship Mechanism)} that again extends greedy to obtain SP mechanism for EFX, EF1, MMS. SD also provides proportional allocations when they exist.
\end{enumerate}
In Table \ref{tab:mvp}, we summarize all the results for the existence of an SP mechanism for various fairness criteria. (Blue ones are the results in this paper.) When the agents are single-minded, SD is a direct SP mechanism that also ensures EF1, EFX, MMS and proportionality when it exists. When the valuations are (additive) identical, RSD is an SP mechanism that ensures EF1. Additive valuations are the most well-studied in literature. For EF, proportionality, EFX, there are counter-examples when there are $2$ agents for proving that an SP mechanism cannot exist when valuations are additive. Even for MMS and EF1 under additive valuations, the results are for $2$ agents.

\section{Preliminaries}
\subsection{Notation}
Consider the problem of division of indivisible resources. We represent each instance by $\langle N, M, V\rangle$ which are formally defined below,
\begin{itemize}[leftmargin=*]
    \item Finite set of agents $N = \{1,\ldots,n\}$
    \item Finite set of indivisible goods $M = \{1, \ldots, m \}$.
    \item Valuation functions $V$ where $v\in V$ denotes a particular profile and $\forall i\in N$,  $v_i: 2^M \rightarrow \mathbb{R}_{+}$. Let $v_{-i}$ be the valuation profile of all agents excluding $i$. 
    \item We assume $v_i$ is monotonic, $\forall i\in N, \forall S \subseteq T \subseteq M, v_i(S) \le v_i(T) $
    \item Additive valuations imply for any $S\subseteq M, v_i(S) = \sum_{j \in S} v_i(\{j\})$
    \item Identical valuations imply $\forall i, j \in N, \forall S\in M, v_i(S) = v_j(S)$. Identical additive valuations imply $v_i$ is both identical and additive.
    \item Single minded agents with desirable bundles $D = (D_1, \ldots, D_m)$. The valuation of an agent $i \in N$ is given by, for a $c \in \mathbb{R}_{+}$
    \begin{equation}
    \label{eq:sm_Def}
    \forall S \in M,\ v_i(S) = 
\begin{cases}
    c,& \text{if } S\supseteq D_i\\
    0,              & \text{otherwise}
\end{cases}
\end{equation}
    \item The set of all possible complete allocations, $\mathcal{A}$. Given $A \in \mathcal{A}$ denotes a specific allocation and $A_i$ is allocation per agent. By complete allocation we mean if there are $m$ goods then $\forall A, \sum_i |A_i| = m$, assuming each resource can be allocated only to a single agent.
 \end{itemize}
 \subsection{Important Definitions}
 We define the relevant fairness notions below with examples,
 \begin{definition}[Proportionality]
Given an instance $\langle N, M, V \rangle$, the allocation $A$ is proportional \emph{iff} $\forall i \in N$,
$$v_i(A_i) \geq \frac{1}{n} v_i(M) $$
 \end{definition}
\subsubsection{Example.} Consider two agents $1$ and $2$ and three goods $a, b, c$. $v_i$ is given below where $(x, y) \in \{ (a,b), (b,c), (c, a) \}$.
\begin{table}[!htb]
    \centering
    \begin{tabular}{cccccc}
    \hline
        & $v(a)$ & $v(b)$ & $v(c)$ & $v(x,y)$ &$v(a,b,c)$ \\
    \hline
      1 & 10 & 20 & 15 & 30 & 30 \\
      2 & 10 & 20 & 15 & 30 & 30\\
    \hline
    \end{tabular}
\end{table}
Possible proportional allocations are when 1 receives item $b$ and 2 receives goods $\{a,c\}$ or vice versa. Also when agent 1 receives $c$ and agent 2 receives goods $\{a,b\}$ or vice versa.
\begin{definition}[Envy-freeness (EF)]
 For $\langle N, M, V \rangle$ an allocation $A$ is envy-free \emph{iff}, 
$$\forall i, j \in N\ \  v_i(A_i) \geq v_i(A_j) $$
\end{definition}

\subsubsection{Example.} \label{ex:ef} Consider 2 agents $1$ and $2$, two goods $a$, $b$. For agent 1, $v_1(a) = 20, v_1(b) = 10$ and for agent 2, $v_2(a) = 10, v_2(b) = 20$. It is envy-free to allocate $a$ to agent 1 and $b$ to agent 2.

Both the notions of proportionality and EF are too strong in the case of indivisible goods and are not guaranteed to exist. Consider the case when there are two agents and only one item, it is impossible to have any allocation that is either EF or even proportional. When the valuations are sub-additive, every EF allocation is proportional as shown in Figure \ref{fig:img1}. In \cite{lipton2004}, the authors define the following notion weaker than EF.

\begin{definition}[\textbf{EF1}] For $\langle N, M, V \rangle$ an allocation $A$ is EF1 \emph{iff} $\forall i,j \in N,$ $\exists a\in A_j$ such that, $$v_i(A_i) \geq v_i(A_j \backslash \{a\})$$
\end{definition}{}
EF1 allocation always exists for general monotone valuations. Another relaxation of EF stronger than EF1 is defined below,

\begin{definition}[\textbf{EFX}] For $\langle N, M, V \rangle$ an allocation $A$ is EFX \emph{iff}, $\forall i,j \in N,$ $\forall a\in A_j$ such that, $$v_i(A_i) \geq v_i(A_j \backslash \{a\})$$
\label{def:efx}
\end{definition}{}

\subsubsection{Example.} We saw before in Example \ref{ex:ef} that EF allocation is not possible when there are two agents and only one item. But an allocation where the item is assigned to either 1 or 2 is both EF1 and EFX. 

Unlike EF1, EFX is guaranteed to exist only for three agents or when agents have identical valuations. The relations between EF, EFX and EF1 is represented in Figure \ref{fig:img1} for any general monotonic valuations.

In \cite{budish2011}, the author defines another threshold based definition of fairness where each agent is guaranteed at least as much valuation as a risk-averse agent would guarantee itself. By risk-averse we mean an agent who assumes that given a partition of the bundles it might end up in receiving the bundle with minimum valuation. Hence, if the agent decides the partition, it would do so to maximize the value of the minimum bundle. This is in the same spirit of a cut-and-choose protocol.
\begin{definition}[\textbf{Maximin Share} (MMS)]
For $\langle N, M, V \rangle$ an allocation $A$ is MMS \emph{iff}  $\forall i\in N$,  $$v_i(S_i) \geq \mu_i$$ where
$$ \mu_i = \max_{A \in \Pi_n(M)} \min_{A_j \in A} v_i(A_j)$$
\label{def:mms}
\end{definition}{}

\begin{figure}[!t]
    \centering
    \includegraphics[width=0.8\columnwidth]{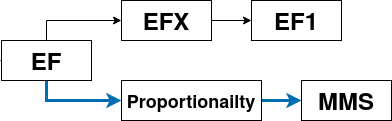}
    \caption{Relation between Various Fairness Criteria \cite{bouveret16}}
    \label{fig:img1}
\end{figure}

\subsubsection{Example.} Consider there are two agents $1$ and $2$, three goods $a, b, c$. We consider additive valuations for both agents. Let the valuation of each item be given as follows,
\begin{table}[!h]
    \centering
    \begin{tabular}{cccc}
    \hline
        & $v(a)$ & $v(b)$ & $v(c)$  \\
    \hline
      1 & 10 & 20 & 40 \\
      2 & 10 & 40 & 20 \\
    \hline
    \end{tabular}
\end{table}
In this example $\mu_1 = \mu_2 = 30$. Agent 1 gets $c$ and 2 gets $a,b$ would be an MMS allocation. 

There is no straight forward relation between EF1/EFX and MMS allocations; for two agents MMS implies EFX. Although any allocation which is  proportional is also MMS when the valuations are sub-additive, Figure \ref{fig:img1}. There exists an MMS allocation for 2 agents but it may not exist for more than two agents.

\subsection{Strategy-Proof Mechanisms}
In mechanism design, we assume the agents are self-interested and strategic. The agents have private information (valuation over goods) that is indispensable for the desired outcome. The agents may or may not reveal their private information based on their individual utility. Mechanism design deals with the two-fold problem of i) \emph{Preference Elicitation} and ii) \emph{Preference Aggregation}. In the former, one explores the specific mechanism in which the agents' best interest lies in revealing their true valuations. The latter, nonetheless challenging, is the problem of obtaining the desired outcome, once the true valuations are known. In our case, this would be finding the fair allocation. There are two kinds of approaches for solving the problem of preference elicitation. 1)Direct Mechanism 2) Indirect Mechanism. We focus on direct mechanism which is defined as follows,
\begin{definition}[Direct Mechanism]
The direct mechanism ($\mathbb{M}$) maps true valuations of the agents to the desired outcome. It is a mapping from the valuations of the agents to the space of allocations $\mathbb{M}: V \rightarrow \mathcal{A}$. 
\end{definition}

 A direct mechanism is strategy-proof (SP) if the agents do not have any incentive to misreport, more formally,
\begin{definition}[Deterministic SP Mechanism]
A deterministic mechanism $\mathbb{M}$ is strategy-proof (SP), if  $\forall v$, $\forall i \in N$, 
$$v_i(\mathbb{M}(v_i, v_{-i})) \geq v_i(\mathbb{M}(v'_i, v_{-i} )), \ \forall v'_i \ \forall v_{-i}$$
where $v'_i$ is a misreported valuation. 
\end{definition}
We also look for little weaker mechanism, in the context of identical valuations. 
\begin{definition}[Deterministic NSP Mechanism]
A deterministic mechanism $\mathbb{M}$ is Nash strategy-proof (NSP), if  $\forall v$, $\forall i \in N$ when other agents report truthfully, 
$$v_i(\mathbb{M}(v_i, v_{-i})) \geq v_i(\mathbb{M}(v'_i, v_{-i} )), \ \forall v'_i$$
where $v'_i$ is a misreported valuation.
\end{definition}

Note that, in NSP, it is best response to each agent to report truthfully if others are reporting truthfully. SP is stronger notion of truthfulness -- no matter what others are reporting, it is a best response for each agent to report truthfully. 

\begin{definition}[strategy-proof Fair Mechanism (SPF)]
A mechanism $\mathbb{M}$ is SPF \emph{iff} $\mathbb{M}$ is SP or NSP and fair (for the given fairness condition).
\end{definition}
Before we discuss the existence of SPF mechanisms for various fairness criteria, we would like to state an observation that makes our search easier. The observation is based on the relationship between the various fairness criteria given in Figure \ref{fig:img1},
\begin{obs}
For any two fairness criteria $X$ and $Y$, if $X \implies Y$ from Figure \ref{fig:img1}, i.e., every allocation that satisfies $X$ also satisfies $Y$. We can conclude that, if there does not exist an SP mechanism for $Y$, then there will not exist an SP mechanism for $X$.
\label{obs:1}
\end{obs}

With this background, we first present impossibilities of SP and fair mechanisms for additive valuations.

\section{Impossibilities of Strategy-Proof and Fair Mechanisms for Additive Valuations.}
As we have discussed before, fair allocations may still cause unrest among agents if some agents choosing to lie benefit when agents adhering to the rules and revealing their true valuations forgo the benefits they could have received. In this paper we are concerned about the existence of truthful mechanisms that can implement the fairness definitions defined above. 

\subsection{EF}
In \cite{lipton2004}, the authors raise the question of the existence of truthful mechanisms that implement EF. 
\begin{theorem} \cite{lipton2004}
Any mechanism that returns an allocation with minimum possible envy cannot be truthful. The same is true for any mechanism that returns an envy-free allocation whenever there exists one.
\label{thm:ef1}
\end{theorem}
As a proof, the authors provide an example consisting of two agents with additive valuation function, where every possible envy-free allocation can be manipulated by either of the agents. An SP mechanism for EF, ($\mathbb{M}^{EF}$)  would select from $\mathcal{A}^{EF}$, i.e, a set of all EF allocations. If there exists a valuation profile $v$, where $\mathcal{A}^{EF}_v$ be all possible EF allocations for $v$, $\forall A^{EF} \in \mathcal{A}^{EF}_v$ and with strict inequality for atleast one $A^{EF}$. 

\begin{equation}
     \exists i,\ \exists v'_i \ s.t.,\ \  v_i(\tilde{A}^{EF}_i) \geq v_i(A^{EF}_i), \quad \forall \tilde{A}^{EF} \in \mathcal{A}^{EF}_{v'}
     \label{eq:1}
\end{equation}
We know that for any deterministic SP mechanism that ensures EF, $\mathbb{M}^{EF}(v) \subseteq \mathcal{A}^{EF}_v$, hence the Equation \ref{eq:1} holds for $\mathbb{M}^{EF}(v)$ which implies that no matter the mechanism, it is always manipulable by certain agent $i$ under the valuation profile, $v$.

\subsection{Proportionality}
For sub-additive valuations, proportionality is a stronger property than MMS (Figure \ref{fig:img1}). In \cite{amanatidis2016}, the authors prove that for 2 agents there is no SP mechanism that ensures better than $\frac{1}{m/2}$-MMS allocation. Hence, it is impossible to have SP mechanism which ensure MMS and hence proportionality for 2 agents. We prove the same by constructing an example guided by Equation \ref{eq:1}.

\bigskip
\noindent Example. Consider $n=2$, $m=3$, we have agents $\{ 1, 2\}$ and goods, $\{a,b,c\}$. The true valuations $v$ are given by Table \ref{tab:p1}. For truthful reporting, there are 2 possible proportional allocations $\mathcal{A}^{prop}_v$ given by Table \ref{tab:p2}. For the first allocation $A^I$, agent 1 obtains a value of $20$. If $1$ reports $v'_i$ as given in Table \ref{tab:p3}, then the only possible proportional allocation is given in Table \ref{tab:p4}, the value for which is $30$, is strictly better than what she was offered. Similarly for next allocation $A^{II}$, agent 2 has an incentive to misreport.

\begin{table}[!htb]
	\centering
	\begin{subtable}[t]{2in}
		\centering
	    \begin{tabular}{cccc}
        \hline
         & $v(a)$ & $v(b)$ & $v(c)$  \\
        \hline
          1 & 20 & 10 & 5  \\
          2 & 5 & 10 & 20  \\
        \hline
        \end{tabular}
        \caption{The true values}\label{tab:p1}
	\end{subtable}
	\quad
	 \begin{subtable}[t]{1in}
	 \centering
    	\begin{tabular}{c|cc}
        \hline
         & 1 & 2 \\
        \hline
         $A^I$ &  a & bc \\
         $A^{II}$ &  ab & c \\
        \hline
        \end{tabular}
        \caption{$\mathcal{A}^{prop}_v$}\label{tab:p2}
   \end{subtable}
	\begin{subtable}[t]{2in}
		\centering
		\begin{tabular}{cccc}
        \hline
         & $v(a)$ & $v(b)$ & $v(c)$   \\
        \hline
          1 & 10 & 10 & 10  \\
          2 & 5 & 10 & 20  \\
        \hline
        \end{tabular}
        \caption{Agent 1 misreports}\label{tab:p3}
	\end{subtable}
	\begin{subtable}[t]{1in}
	\centering
    	\begin{tabular}{c|cc}
        \hline
        & $1$ & $2$ \\
        \hline
         $\tilde{A}^I$ & ab & c \\
        \hline
        \end{tabular}
        \caption{$\mathcal{A}^{prop}_{v'}$}\label{tab:p4}
	\end{subtable}
	\caption{Counter Example for Proportionality}\label{tab:p}
	
\end{table}

\subsection{EFX}
In \cite{amanatidis2017}, they prove that it is impossible to design SP mechanism for EF1 for $n=2$ and $m\geq 5$. Since EFX (Definition \ref{def:efx}) is a stronger property the same result holds when $m>=5$. We prove that it is also impossible to have an SP mechanism for EFX when $m=4$.
\begin{theorem}
Any mechanism that returns an allocation that is EFX cannot be truthful even in the case of additive valuations.
\label{thm:efx}
\end{theorem}
\begin{proof}
Consider an example where $n=2$, we have agents $\{1,2 \}$ and $m=4$, $\{a,b,c,d\}$. For truthful reporting $v$ as given in Table \ref{tab:e1}, there are 4 possible EFX allocations $\mathcal{A}^{EF}_v$ given by Table \ref{tab:e2}. For the first two allocations $A^I, A^{II}$, agent 1 receives $b$ which it values at $100$ or $\{b,c\}$ which it values $120$. If agent $1$ reports $v'_i$ as given in Table \ref{tab:e3}, where the total valuation is same i.e., 200. Under this misreport, the only possible EFX allocations are given in Table \ref{tab:e4} in which agent 1 receives at least $140$ which is at least as good as the value it received for truthful reporting. Similarly for next two allocations $A^{III}, A^{IV}$ agent 2 has an incentive to misreport. Hence there are only four possible EFX allocations and for each allocation at least one agent has an incentive to misreport. 
\begin{table}[!htb]
	\centering
	\begin{subtable}[t]{2in}
		\centering
	    \begin{tabular}{ccccc}
        \hline
         & $v(a)$ & $v(b)$ & $v(c)$ & $v(d)$ \\
        \hline
          1 & 40 & 100 & 20 & 40 \\
          2 & 100 & 40 & 20 & 40 \\
        \hline
        \end{tabular}
        \caption{The true values}\label{tab:e1}
	\end{subtable}
	\quad
	 \begin{subtable}[t]{1in}
	 \centering
        \begin{tabular}{c|cc}
        \hline
         & $1$ & $2$ \\
        \hline
         $A^I$ &  b & acd \\
         $A^{II}$ &  bc & ad \\
         $A^{III}$ &  bd & ac \\
         $A^{IV}$ &  bdc & a \\
        \hline
        \end{tabular}   
        \caption{$\mathcal{A}^{EFX}_v$}\label{tab:e2}
   \end{subtable}
	\begin{subtable}[t]{2in}
		\centering
    	\begin{tabular}{ccccc}
        \hline
         & $v(a)$ & $v(b)$ & $v(c)$ & $v(d)$ \\
        \hline
          1 & 90 & 70 & 15 & 25 \\
          2 & 100 & 40 & 20 & 40 \\
        \hline
        \end{tabular}
        \caption{Agent 1 misreports}\label{tab:e3}
	\end{subtable}
	\begin{subtable}[t]{1in}
	\centering
        \begin{tabular}{c|cc}
        \hline
         & $1$ & $2$ \\
        \hline
         $\tilde{A}^I$  & bd & ac \\
          $\tilde{A}^I$ & bdc & a \\
        \hline
        \end{tabular}
        \caption{$\mathcal{A}^{EFX}_{v'}$}\label{tab:e4}
	\end{subtable}
	\caption{Counter Example for EFX}\label{tab:e}
	
\end{table}
\end{proof}

With this we can conclude, that under additive valuations, there is an instance where no SP mechanism can be EFX.

\subsection{EF1}
In this subsection, we explore the existing algorithms that find EF1 allocations and prove that these are manipulable. We provide an instance with $n=2$ for each case.
\subsubsection{Greedy Round-Robin Algorithm}
In \cite{cara16}, the authors provide a simple algorithm for obtaining EF1 allocations when the valuations are additive. It involves the following steps,
\begin{itemize}
    \item Fix an arbitrary order on the agents
    \item Allocate the first agent its most valuable good
    \item The next agent is allocated its most valuable among the remaining goods
    \item The algorithm terminates when all the goods are allocated
\end{itemize}
The following example shows that the above algorithm can be manipulated by the agents.
\begin{proposition}
Greedy round-robin algorithm is manipulable for additive valuations. 
\label{prop:greedy}
\end{proposition}
\begin{table}[!htb]
	\centering
	\begin{subtable}[t]{2in}
		\centering
	    \begin{tabular}{cccccc}
        \hline
         & $v(a)$ & $v(b)$ & $v(c)$ & $v(d)$ & $v(e)$\\
        \hline
          1 & 12 & 10 & 8 & 6 & 1 \\
          2 & 1 & 10 & 8 & 6 & 9\\
        \hline
        \end{tabular}
        \caption{The true values}\label{tab:g1}
	\end{subtable}
	\quad
		\begin{subtable}[t]{1in}
	\centering
        \begin{tabular}{c|cc}
        \hline
         & $1$ & $2$ \\
        \hline
         $1 \rightarrow 2$  & acd & be \\
          $2 \rightarrow 1$ & ac & bde \\
        \hline
        \end{tabular}
        \caption{EF1 allocations}\label{tab:g2}
	\end{subtable}
	 \begin{subtable}[t]{2in}
	 \centering
         \begin{tabular}{cccccc}
        \hline
         & $v(a)$ & $v(b)$ & $v(c)$ & $v(d)$ & $v(e)$ \\
        \hline
          1 & 10 & 12 & 8 & 6 & 1\\
          2 & 1 & 10 & 8 & 6 & 9\\
        \hline
        \end{tabular}   
         \caption{Agent 1 misreports}\label{tab:g3}
   \end{subtable}
   \quad
   		\begin{subtable}[t]{1in}
	    \centering
        \begin{tabular}{c|cc}
        \hline
         & $1$ & $2$ \\
        \hline
         $1 \rightarrow 2$  & abd & ce \\
        \hline
        \end{tabular}
        \caption{EF1 allocation}\label{tab:g4}
	\end{subtable}
   
	\begin{subtable}[t]{2in}
		\centering
    	\begin{tabular}{cccccc}
        \hline
         & $v(a)$ & $v(b)$ & $v(c)$ & $v(d)$ & $v(e)$\\
        \hline
          1 & 12 & 10 & 8 & 6 & 1 \\
          2 & 1 & 10 & 8 & 8 & 5 \\
        \hline
        \end{tabular}
        \caption{Agent 2 misreports}\label{tab:g5}
	\end{subtable}
	  \quad
   		\begin{subtable}[t]{1in}
	    \centering
        \begin{tabular}{c|cc}
        \hline
         & $1$ & $2$ \\
        \hline
         $2 \rightarrow 1$  & ad & bce \\
        \hline
        \end{tabular}
        \caption{EF1 allocation}\label{tab:g6}
	\end{subtable}
	\caption{Greedy round-robin is manipulable}\label{tab:greed}
	
\end{table}
\begin{proof}
 Example. Consider $n=2$, $\{1,2\}$, $m=5$ , $\{a,b,c,d, e\}$ where the valuations are additive and  given by Table \ref{tab:g1}. When $n=2$, there is only two possible orders among the agents. When applying greedy algorithm with 1 followed by 2 or $1\rightarrow 2$, agent 1 gets $\{a,c,d\}$ (Table \ref{tab:g2}) with a value of $26$. If the agent misreports its value as given in Table \ref{tab:g3}, the allocation that agent 1 gets is $\{a,b,d \}$ which it values at $28$ that is strictly more than when it was truthful. Similarly agent 2 can misreport to an advantage when the order is $2 \rightarrow 1$. It improves its allocation from $\{ b,d,e\}$ (Table \ref{tab:g2}) that it values at $25$  to $\{b,c,e \}$ whose value is $27$.   
\end{proof}

\subsubsection{Cycle-elimination Algorithm}
Greedy method fails for general valuations, instead the cycle-elimination algorithm \cite{lipton2004} provides EF1 solution in polynomial time in general. The algorithm is as follows,
\begin{itemize}
    \item Goods are allocated in arbitrary order
    \item An envy-graph is maintained where the agents are the vertices and a directed edge $i \rightarrow j$ represents that agent $i$ envies agent $j$ under the current allocation.
    \item The next item is allocated to the agent with no incoming edge. If there is a cycle, it can be eliminated by exchanging the goods of the agents that form the cycle, with the ones they envy.
\end{itemize}
We show that the above algorithm is manipulable by the agents.
\begin{proposition}
Cycle-elimination algorithm is manipulable even for identical valuations.
\end{proposition}
\begin{table}[!htb]
	\centering
	\begin{subtable}[t]{2in}
		\centering
	    \begin{tabular}{ccc}
        \hline
          1 & 2 & graph \\
        \hline
          d &  & $1 \leftarrow 2$ \\
          d & a & $1 \leftarrow 2$ \\
          d & ab & $1 \rightarrow 2$ \\
          dc & ab & $1 \rightarrow 2$ \\
        \hline
        \end{tabular}
        \caption{Cycle-elimination on $v$}\label{tab:c1}
	\end{subtable}
		\begin{subtable}[t]{2in}
	\centering
       \begin{tabular}{ccc}
        \hline
          1 & 2 & graph \\
        \hline
          d &  & $1 \leftarrow 2$ \\
          d & a & $1 \rightleftharpoons 2$ \\
          a & d & no envy \\
          ab & d & $1 \leftarrow 2$ \\
          ab & dc & $1 \leftarrow 2$ \\
        \hline
        \end{tabular}
        \caption{Cycle-elimination on $v'$}\label{tab:c2}
	\end{subtable}
	\begin{subtable}[t]{1in}
       \begin{tabular}{ccc}
        \hline
          1 & 2 & graph \\
        \hline
          d &  & $1 \leftarrow 2$ \\
          d & a & $1 \rightleftharpoons 2$ \\
          a & d & no envy \\
          a & bd & $1 \rightarrow 2$ \\
          ac & bd & $1 \leftarrow 2$ \\
        \hline
        \end{tabular}
        \caption{Cycle-elimination on $v'$}\label{tab:c3}
	\end{subtable}
	\caption{Cycle-elimination is manipulable}\label{tab:cyc}
\end{table}

Consider $n=2$, $\{1,2\}$, and $m=4$, $\{a,b,c,d \}$. Let $(x,y) \in \{(a,b) (b,c), (a,c)\}$ where the valuations $v$ of the agents are identical and given below. The value of other subsets not mentioned below are additive.
\begin{table}[!htb]
    \centering
    \begin{tabular}{ccccc}
    \hline
     $v(a)$ & $v(b)$ & $v(c)$ & $v(d)$ &$v(x, y)$  \\
    \hline
        5 & 5 & 5 & 10 & 16 \\
    \hline
    \end{tabular}
\end{table}
When we run the cycle-elimination algorithm, the steps are as given in the Table \ref{tab:c1}. It is easy to see that the agent who gets the item $d$ (w.l.o.g we assume agent $1$ gets $d$) always ends up with a value $15$ and can try to increase the utility by gaining the other bundle whose value is $16$. Now consider the following misreported valuation by the agent who gets item $d$. 
\begin{table}[!htb]
    \centering
    \begin{tabular}{ccccc}
    \hline
     $v'(a)$ & $v'(b)$ & $v'(c)$ & $v'(d)$ &$v'(x, y)$  \\
    \hline
        5 & 5 & 5 & 4 & 16 \\
    \hline
    \end{tabular}
\end{table}
With the misreported valuation, we again run the cycle-elimination algorithm and there can be two possible outcomes as presented in Table \ref{tab:c2}, \ref{tab:c3}. We see that in these cases the agent $1$ receives $\{a,b \}$ or $\{ a, c\}$ which it values at $16$ i.e., strictly more than the previous value for $\{ d,c\}$ that is 15. In this example we prove that there is always an agent that can manipulate to increase its utility. 

In the above case, we considered an example for (general) identical valuations. In fact it also possible to manipulate cycle-elimination for (additive) identical valuations.

\noindent Example. Consider $n=2$ and $m=3$ where the agents have the following (additive) identical valuation $v$ and the agent that does not receive the good $c$ misreports the valuation to $v'$ also given below.
\begin{table}[!htb]
    \centering
    \begin{tabular}{ccccc}
     \hline
     $v(a)$ & $v(b)$ & $v(c)$   \\
    \hline
        5 & 5 & 12 \\
    \hline
     $v'(a)$ & $v'(b)$ & $v'(c)$   \\
    \hline
        4 & 5 & 12 \\
    \hline
    \end{tabular}
\end{table}
For many orderings over $m$ that the algorithm chooses, the value obtained for $v'$ is as good as $v$. But when the ordering is chosen to be $a$ then $b$ then $c$ or $(b, a, c)$, the agent manipulating ensures a value of $17$ as opposed to just receiving $5$.  
\section{Identical Additive Valuations}
When valuations are identically additive, we know that picking sequences are strategy-proof \cite{lang11}. Based on picking sequences, we provide an algorithm, (RSD) \emph{(Repeated Serial Dictatorship)} to obtain truthfulness while ensuring EF1.
\begin{algorithm}[!t]

 \SetKwData{Left}{left}\SetKwData{This}{this}\SetKwData{Up}{up}
 \SetKwFunction{Union}{Union}\SetKwFunction{FindCompress}{FindCompress}
 \SetKwInOut{Input}{Input}\SetKwInOut{Output}{Output}

 \Input{ $\langle N, M, V \rangle$, $V$ is identical additive}
 \Output{$(A_1, A_2, \ldots, A_n) \in \mathcal{A}^{EF1}$}
 \BlankLine

 Set an arbitrary but fixed order on the agents, w.l.o.g, $(1, 2, \ldots, n)$  \;
 $A_i = \phi, \quad \forall i$\;
 $R = M$ (goods remaining after each iteration) \;
 $i =0$ (agent number)\;
 \While{$R \neq \phi$}{
  $x \in \underset{j \in R}{argmax} \ v_i(j) $ \;
  $A_i = A_i \bigcup x$ \;
  $R = R \setminus x$ \;
  $i = (i+1) \mod n$ \;
 }
 \caption{Repeated Serial Dictatorship Mechanism (RSD)}\label{algo:rsd}
\end{algorithm}

 Repeated Serial Dictatorship is EF1 and (i) SP when $m \leq n$ and (ii) NSP when the valuations are identical and additive. Given that RSD implements greedy round-robin algorithm under additive valuations, the output allocation $A$ is EF1.

\noindent \underline{\textbf{(i) Case $m \leq n$}}: under this case, the while loop in Algorithm \ref{algo:rsd} runs for $m$-iterations, given $m \leq n$, each agent only gets one chance to participate and select the item $x$. The ordering chosen by the algorithm is independent of the agent valuations hence cannot be manipulated. From the algorithm we know that given the remaining goods $R$,
$$ A_i = x \in \underset{j \in R}{argmax} \ v_i(j)  $$
If the agent misreports s.t. $y \in  \underset{j \in R}{argmax} \ v'_i(j)$ and $x \neq y$. The agent receives $y$ s.t. under true valuations, $ v_i(y) \leq v_i(x)$ and hence cannot strictly increase its utility.

\noindent \underline{\textbf{(ii) Case} $v$ is Identical (Additive)}: Let $v$ be the truthful report, we assume $v^1 \geq v^2 \geq \ldots \geq v^m$ be the value all the agents have for the $m$ goods in decreasing order. The Algorithm \ref{algo:rsd} will continue for $m$ rounds and assign the goods in this order itself. The goods remaining at round $j$ is given by $ R_j = \{ v^j, \ldots, v^m\}$.
Let us assume an agent $i$ gets allocated $k$ items before the algorithm terminates, then it selects from the following subsets and receives the items it values the most in each of these,
 $$ \{ R_i, R_{i+n},\ldots, R_{i+kn} \}$$
 Hence $ A_i = \{v^i, v^{i+n}, \ldots, v^{i+kn}\}$.
If the agent $i$ misreports and the remaining agents report truthfully, in any of the rounds w.l.o.g, $i^{th}$ round s.t., the relative ordering between the items changes, then the agent might face the two possible sets in the next round ($i+n$),
\begin{itemize}[leftmargin=*]
    \item Misreport s.t. agent $i$ gets item $p$ instead of $i$, $p \leq n+i-1$, then the set it faces in the next rounds is $\{ R_{n+i}, \ldots, R_{i+kn} \}$. Hence the items allocated are $A'_i = \{ v^p, v^{i+n}, \ldots, v^{i + kn}\}$. It can be clearly verified that, $v_i(A_i) \geq v_i(A'_i)$, hence no incentive to misreport.
    
    \item Misreport s.t agent $i$ gets item $p$ where, $k'n + i > p \geq (k'-1)n+i$, $k'\geq 2$ then the sets $i$ faces are $$\{ R_{n+i-1} \setminus \{p\}, \ldots, R_{i + (k'-1)n-1} \setminus \{p\}, R_{i+k'n}, \ldots, R_{i+kn} \}$$ Hence the items allocated are $$A''_i = \{ v^p, v^{n+i-1}, \ldots,v^{i+(k'-1)n-1} ,v^{i+k'n}, \ldots,v^{i+kn}\} $$
    Using the fact that $p \geq i+(k' -1)n$, we know that $v^{i +(k'-1)n } \geq v^p$. Hence we compare the sets $A_i$ and $A''_i$ ($\succeq$ represents element-wise comparison) as follows to obtain $v_i(A_i) \geq v_i(A''_i)$, 
    \begin{equation*}
        \begin{split}
            \{v^i, v^{i+n}, &\ldots, v^{i +(k'-2)n }, v^{i +(k'-1)n }, v^{i + k'n},\ldots, v^{i + kn}\} \succeq \\
            & \{v^{i+n-1},v^{i+2n-1},\ldots, v^{i +(k'-1)n},v^p,v^{i+k'n},\ldots, v^{i+kn}\}\\
        \end{split}
    \end{equation*}
\end{itemize}
This completes the proof for truthfulness for RSD under additive and identical valuations.

\section{Single-Minded Agents}

In this section, we restrict to a simpler valuation profile. We assume the agents are (SM) single minded. SM agents are only interested in a single bundle of goods $D$. Upon receiving the specific bundle or any super-set they get a positive utility and zero value for any other bundle (Formally given by Equation \ref{eq:sm_Def}). The problem instance is denoted by $\langle N, M, D \rangle$.


\begin{obs}
When all the agents are SM, any allocation is MMS and EF1. The $\mu_i$ in Definition \ref{def:mms} is 0 in this setting when $n>1$. Hence, allocating all the goods to one agent is also MMS. Similarly, all possible allocations satisfy EF1. If an agent $i$ receives its desired bundle or super-set then it doesn't envy any agent. If an agent $j$ receives $D_i$, then removing any item $x \in D_i$ would remove envy. If no agent receives $D_i$ as a whole, there is no envy. 
\label{obs:f}
\end{obs}

Based on the observation, we extend greedy round-robin algorithm to design a SD \emph{(Serial Dictatorship)} mechanism is SP since it is also a picking sequence. SD trivially satisfies EF1 and MMS and we prove that it also satisfies EFX.

\begin{algorithm}[!t]

 \SetKwData{Left}{left}\SetKwData{This}{this}\SetKwData{Up}{up}
 \SetKwFunction{Union}{Union}\SetKwFunction{FindCompress}{FindCompress}
 \SetKwInOut{Input}{Input}\SetKwInOut{Output}{Output}

 \Input{ $\langle N, M, D \rangle$, $D = (D_1, D_2, \ldots, D_n)$}
 \Output{$(A_1, A_2, \ldots, A_n)$}
 \BlankLine

 Order the agents s.t. $|D_1|\leq |D_2| \leq \ldots \leq |D_n| $  (Ties broken arbitrarily)\;
 $i =0$ (agent number)\;
 $R = M$ (goods remaining after each iteration) \;
 \While{$R \neq \phi$}{
  Let $D_i$ be the preferred set for the current agent $i$\;
  \eIf{$D_i \subseteq R$}{
   $A_i = D_i$\;
   $R = R \setminus D_i$, $i=i+1$ \;
   }{
   \eIf{$i < n$}{
   $i = i+1$;
   }{
   $A_i = R$\;
   $R = \phi$\;
   }
  }
 }
 \caption{Serial Dictatorship Mechanism (SD)}\label{algo:sd}
\end{algorithm}
We prove certain desirable properties of the Algorithm \ref{algo:sd}, which is a modified version of greedy round-robin algorithm.
\begin{theorem}
 The Serial Dictatorship Mechanism is strategy-proof (SP) and also satisfies EF1, MMS and EFX when the agents are single-minded.
\end{theorem}    
\begin{proof}
In the Algorithm \ref{algo:sd}, the while loop can run for a maximum of $n$ rounds. This means each agent $i$ has only one round in which it can be allocated the preferred bundle $D_i$. The ordering is according to the increasing cardinality of $D_i$. An agent can manipulate the ordering by reporting its desired bundle as  $D'_i$ s.t., $|D'_i| < |D_i|$. This means the agent will be allocated if at all a bundle that it does not desire. If  $|D'_i| > |D_i|$ then the probability that the agent gets any allocation is strictly less than when it reports truthfully. Hence the agent does not have any incentive to manipulate the ordering.

Given that the agent cannot manipulate the ordering. At any round, it is optimal for the agent to report truthfully the desired set $D_i$. Now we prove that the allocation $A$ obtained from SD satisfy the following fairness criteria,
\begin{itemize}[leftmargin=*]
    \item (EF1 and MMS). This is trivially true due to  Observation \ref{obs:f} which states that any allocation is EF1 and MMS when we have SM agents.
    \item (EFX). Let us assume $k$ agents, denoted by $L$ (lucky), are allocated their desired sets hence do not have any envy. If $ k \neq n$, then $n-k$ agents, denoted by $U$ (unlucky), did not receive their desired subset. From the algorithm we know that for any agent $i \in U$, $D_i \not\subseteq R_i$ where $R_i$ is the set of goods at the beginning of $i^{th}$ round.   
    \begin{itemize}
        \item $\forall i, j \in U$, $i$ does not envy $j$, because $j$ is allocated empty bundle unless $j$ is the agent appearing at the last $n$ and receives the items remaining. In this case since agent $i$ is given the chance to chose before $j$ which clearly shows it cannot envy $j$.
        \item $\forall i \in U$, $\forall \bar{i} \in L$, if $\bar{i} < i$, then $|D_{\bar{i}}| < |D_{i}|$, hence agent $i$ cannot envy $\bar{i}$. If $|D_{\bar{i}}| = |D_{i}|$ then removing any item from the bundle of $\bar{i}$ will remove envy. Hence it still satisfies EFX.
        \item $\forall i \in U$, $\forall \bar{i} \in L$, if $\bar{i} > i$, then $D_{\bar{i}} \subseteq R_i $ hence $D_{\bar{i}} \neq D_{i}$, hence the agent $i$ does not envy $\bar{i}$
    \end{itemize}
    Hence the allocation is EFX.
\end{itemize}
This concludes the proof for the theorem. Hence SD is SP and provides allocations that satisfy EF1, MMS and EFX.

\noindent\emph{Note (Proportionality). } When the agents are SM, proportional allocation exists when the following is true,
$$ v_i(A_i) \geq \frac{1}{n} v_i(D_i) > 0, \ \forall i \in N$$
The above is true only when all the agents get their desired bundle. If such a solution exists then it easily found by the SD.
\end{proof}

\section{Future Work and Conclusion}
In the literature, there are many algorithms for finding fair division of resources. Yet such algorithms may not be really fair, if one agent can manipulate it by misreporting its value to obtain higher utility. We show that greedy round-robin and cycle-elimination algorithms are manipulable. In general, we study the possibility of having strategy-proof, deterministic mechanisms without money which ensure various criteria of fairness like EF, proportionality, EFX, EF1, MMS. It is known that, such a mechanism does not exist for EF, proportionality and MMS under additive valuations. It also does not exist for EF1 under additive valuations when the number of items are more than $5$. We prove that it does not exist for EFX even when the number of items are $4$. 

Given these impossibility results, we look into settings where agents have simpler valuation type like single minded bidders. Under this assumption we provide a strategy-proof algorithm SD. SD satisfies all fairness criteria except EF. RSD satisfies EF1. The results are summarized in Table \ref{tab:mvp}. For future work, it would be interesting to look into mechanisms to settle the unfinished components in the table. Given the impossibility for general valuations, it would be interesting to design mechanisms for more specific valuation types for e.g., (general) identical etc.

\bibliographystyle{ACM-Reference-Format}
\bibliography{acmart}

\end{document}